\documentclass[11pt]{amsart}

\usepackage{amsmath}
\usepackage{amssymb}
\usepackage{graphicx}

\usepackage{pict2e}

\usepackage{bbm}

\newtheorem{theorem}{Theorem}[section]
\newtheorem{cor}[theorem]{Corollary}
\newtheorem{prop}[theorem]{Proposition}
\newtheorem{lemma}[theorem]{Lemma}

\theoremstyle{remark}
\newtheorem{remark}[theorem]{Remark}
\newtheorem{defn}[theorem]{Definition}
\newtheorem{example}[theorem]{Example}

\newcommand{\re}{\operatorname{Re}}
\newcommand{\im}{\operatorname{Im}}

\let\text=\mbox

\catcode`\@=11

\def\a{\alpha}
\def\b{\beta}

\def\d{\delta}
\def\g{\lambda}

\def\q{\quad}

\def\eR{{\bf R}}
\def\eN{{\bf N}}
\def\Ze{{\bf Z}}

\def\Ce{{\bf C}}

\def\ty{\infty}
\def\e{\varepsilon}

\def\:{{\penalty10000\hbox{\kern1mm\rm:\kern1mm}\penalty10000}}
\def\ov#1{\overline{#1}}

\def\O{\Omega}
\def\pa{\partial}

\def\st{\subset}
\def\stq{\subseteq}

\newcount\br@j
\def\q{\quad}

\def\bg{\begin}

\def\endeqnn{\end{eqnarray*}}
\def\bgeqn{\bg{eqnarray}}
\def\endeqn{\end{eqnarray}}

\def\bgeqq#1#2{\bgeqn\label{#1} #2\left\{\begin{array}{ll}}
\def\endeqq{\end{array}\right.\endeqn}

\def\eR{\mathbb{R}}
\def\eN{\mathbb{N}}
\def\Ze{\mathbb{Z}}

\def\Ce{\mathbb{C}}

\def\I{\mathbbm{i}}
\def\E{\mathrm{e}}
\def\dd{{\mathrm d}}

\numberwithin{equation}{section}

\title{Essential singularities of fractal zeta functions}

\author[M.\ L.\ Lapidus]{Michel L.\ Lapidus}

\address[M.\ L.\ Lapidus]{University of California, Riverside, Department of Mathematics,
900 University Avenue, Riverside, CA 92521-0135, USA}
\email{{\tt lapidus@math.ucr.edu}}

\author[G.\ Radunovi\'c]{Goran Radunovi\'c}
\address[G.\ Radunovi\'c]{University of Zagreb, Faculty of Science - Department of Mathematics, Bijeni\v cka Cesta 30, 10000 Zagreb, Croatia}
\email{\tt goran.radunovic@math.hr}

\author[D.\ \v Zubrini\'c]{Darko \v Zubrini\'c}
\address[D.\ \v Zubrini\'c]{University of Zagreb, Faculty of Electrical Engineering and Computing, Unska 3, 10000 Zagreb, Croatia}
\email{\tt darko.zubrinic@fer.hr}

\thanks{The work of Michel L.\ Lapidus was partially supported by the US National Science Foundation (NSF) under the research grants DMS-0707524 and DMS-1107750, as well as by the Burton Jones Chair Fund, as he has been the holder of the Burton Jones Chair in Pure Mathematics since June 2017 at the University of California, Riverside.}
\thanks{The research of Goran Radunovi\'c was supported by the Croatian Science Foundation under the project IP-2014-09-2285 and UIP-2017-05-1020.}
\thanks{The research of Darko \v Zubrini\'c was supported by the Croatian Science Foundation under the project IP-2014-09-2285.}



\begin{document}

\begin{abstract}
We study the essential singularities of geometric zeta functions $\zeta_{\mathcal L}$, associated with bounded fractal strings $\mathcal L$. For any three prescribed real numbers $D_\ty$, $D_1$ and $D$ in $[0,1]$, such that $D_\ty<D_1\le D$, we construct a bounded fractal string $\mathcal L$ such that $D_{\rm par}(\zeta_{\mathcal L})=D_{\ty}$, $D_{\rm mer}(\zeta_{\mathcal L})=D_1$ and $D(\zeta_{\mathcal L})=D$. Here, $D(\zeta_{\mathcal L})$ is the abscissa of absolute convergence of $\zeta_{\mathcal L}$, $D_{\rm mer}(\zeta_{\mathcal L})$ is the abscissa of meromorphic continuation of $\zeta_{\mathcal L}$, while $D_{\rm par}(\zeta_{\mathcal L})$ is the infimum of all real numbers $\a$ such that $\zeta_{\mathcal L}$ is holomorphic in the open right half-plane $\{\re s>\a\}$, except for possible isolated singularities  in this half-plane. Defining $\mathcal L$ as the disjoint union of a sequence of suitable generalized Cantor strings, we show that the set of accumulation points of the set $S_\ty$ of essential singularities of $\zeta_{\mathcal L}$, contained in the open right half-plane $\{\re s>D_{\ty}\}$, coincides with the vertical line $\{\re s=D_{\ty}\}$. 
We extend this construction to the case of distance zeta functions $\zeta_A$ of compact sets $A$ in $\eR^N$, for any positive integer $N$.
\end{abstract}

\maketitle


\section{Introduction and notation}

\subsection{Introduction} In the theory of bounded fractal strings, developed since the early 1990s by the first author and his collaborators in numerous papers and several research monographs (see the books \cite{lapidusfrank12,fzf}, the survey article \cite{Lap}, and the many relevant references therein), to each fractal string $\mathcal L$ a set of {\em complex dimensions}, denoted by $\dim_{\Ce} \mathcal L$, is assigned, defined as the set of poles of the corresponding geometric zeta function $\zeta_{\mathcal L}$, suitably meromorphically extended. In this paper, we provide a construction of a class of fractal strings such that the corresponding geometric zeta functions generate {\em essential\,} singularities accumulating along a prescribed vertical line $\{\re s=D_{\ty}\}$ of the complex plane, with arbitrarily prescribed $D_{\ty}\in[0,1)$. This is a new phenomenon appearing in the theory of fractal strings. The main result is stated in Theorem~\ref{main}.

The first example of a fractal string $\mathcal L$, the geometric zeta function $\zeta_{\mathcal L}$ of which possesses essential singularities, has been constructed in \cite[Example~3.3.7 on p.~215]{fzf} (see also \cite{ra1}), starting from the classical Cantor string. In this paper, we first extend this construction to a class of generalized Cantor strings depending on two real parameters.
\bigskip

\subsection{Notation}\label{notation} Following \cite{fzf}, we introduce some basic notation that we shall need in the sequel.

A {\em bounded fractal string} $\mathcal L=(\ell_j)_{j\in\eN}$, is defined as being either a nonincreasing infinite sequence of positive real numbers such that $\sum_{j=1}^\ty\ell_j<\ty$ or else a finite sequence of positive real numbers. Its {\em length} is 
\begin{equation}
|\mathcal L|_1:=\sum_{j=1}^\ty\ell_j.
\end{equation}
For any two bounded fractal strings $\mathcal L_1=(\ell_{1j})_{j\in\eN}$ and $\mathcal L_2=(\ell_{2k})_{k\in\eN}$, we define their {\em tensor product},
\begin{equation}
\mathcal L_1\otimes\mathcal L_2:=(\ell_{1j}\ell_{2j})_{j,k\in\eN},
\end{equation}
 as the fractal string consisting of all possible products $\ell_{1j}\ell_{2j}$, where $i,j\in\eN$, counting the multiplicities. It is also bounded, since $|\mathcal L_1\otimes\mathcal L_2|_1=|\mathcal L_1|_1\cdot|\mathcal L_2|_1<\ty$. 
We can also define their {\em disjoint union} $\mathcal L_1\sqcup\mathcal L_2$ as the union of multisets; that is, each element of $\mathcal L_1\sqcup\mathcal L_2$ has the multiplicity equal to the sum of its multiplicities in $\mathcal L_1$ and $\mathcal L_2$. It is possible to define the {\em disjoint union} $\sqcup_{i=1}^\ty {\mathcal L_i}$ of an infinite sequence $\mathcal L_i=(\ell_{ij})_{j\in\eN}$ of bounded fractal strings, where $i\in\eN$, provided $\sum_{i,j}\ell_{ij}<\ty$. 
For any positive real number $\g$ and a bounded fractal string $\mathcal L=(\ell_j)_{j\in\eN}$, we can define a new fractal string $\g\mathcal L:=(\g\ell_j)_{j\in\eN}$.

The {\em geometric zeta function} $\zeta_{\mathcal L}$ of a given bounded fractal string $\mathcal L=(\ell_j)_{j\in\eN}$ is defined by 
\begin{equation}\label{eq:3}
\zeta_{\mathcal L}(s):=\sum_{j=1}^\ty\ell_j^s,
\end{equation}
 where $s$ is a complex number with $\re s>1$. Clearly, $\zeta_{\mathcal L}(1)=|\mathcal L|_1<\ty$. The {\em abscissa of absolute convergence} of $\zeta_{\mathcal L}$ is denoted by $D(\zeta_{\mathcal L})$, while the {\em abscissa of meromorphic continuation} of $\zeta_{\mathcal L}$ is denoted by $D_{\rm mer}(\zeta_{\mathcal L})$. It can be easily verified that $-\ty\le D_{\rm mer}(\zeta_{\mathcal L})\le D(\zeta_{\mathcal L})\le1$, where $D(\zeta_{\mathcal L}):=\inf\{\a\in\eR:\sum_{j=1}^\ty\ell_j^\a<\ty\}$ coincides with the Minkowski dimension, $\dim{\mathcal L}$, of the fractal string whenever the fractal string $\mathcal{L}$ is infinite, i.e., whenever $(\ell_j)_{j\in\eN}$ is an infinite sequence of positive numbers tending to zero.\footnote{Since $\mathcal{L}$ is bounded, we then always have that $0\leq\dim\mathcal{L}\leq 1$ and hence, similarly for $D(\zeta_{\mathcal{L}})=\dim\mathcal{L}$.}
 The notions of abscissa of absolute convergence and of meromorphic continuation can be extended to general Dirichlet-type integrals; see \cite[esp., Appendix A]{fzf} for details.

It is easy to verify that $\zeta_{\g\mathcal L}(s)=\g^s\zeta_{\mathcal L}(s)$ and $\zeta_{\mathcal L_1\otimes\mathcal L_2}(s)=\zeta_{\mathcal L_1}(s)\cdot \zeta_{\mathcal L_1}(s)$, for all $s\in\Ce$ with $\re s>1$; see \cite[Lemma 3.3.2]{fzf}. Furthermore, $\zeta_{\mathcal L_1\sqcup\mathcal L_2}(s)=\zeta_{\mathcal L_1}(s)+\zeta_{\mathcal L_2}(s)$, for all $s\in\Ce$ with $\re s>1$.

For any given real number $\a$, we define the corresponding {\em vertical line} $\{\re s=\a\}:=\{s\in\Ce:\re s=\a\}$ in the complex plane, while the corresponding {\em open right half-plane} $\{s\in\Ce:\re s>\a\}$ is denoted by $\{\re s > \a\}$. For any two real numbers $\a$ and $\b$, we define $\a+\b\I\Ze:=\{\a+\b\I j\in\Ce:j\in\Ze\}$, which is an arithmetic set contained in the vertical line $\{\re s=\a\}$ of the complex plane.
Here and thereafter, we let $\I:=\sqrt{-1}$ denote ``the'' complex square root of $-1$.

\begin{remark}[{{\em Geometric realization of bounded fractal strings}}]
	\label{rem:1.1}
	A natural way in which bounded fractal strings arise is as follows (see \cite{lapidusfrank12} and the earlier references).
	Consider an open set $\O$ of $\eR$, with boundary denoted by $\partial\O$ and with finite length (i.e., one-dimensional Lebesgue measure) $|\Omega|_{1}$.\footnote{The boundary $\partial\O\subseteq\eR$ is always compact and, in the applications, is often a ``fractal'' subset of $\eR$; see, e.g., Example \ref{ex:1.2}, where $\partial\O$ is the classic (ternary) Cantor set.}
	Then, $\O=\bigcup_{j\geq 1}I_j$, where the (finite or countable) family $(I_j)_{j\geq 1}$ consists of bounded open intervals $I_j$ of lengths $\ell_j$.
	These intervals are simply the connected components of the open set $\O$.
	Without loss of generality and since $|\O|_1=\sum_{j\ge1} \ell_j<\infty$ (because the fractal string $\mathcal{L}:=(\ell_j)_{j\in\eN}$ is bounded), one may assume that $(\ell_j)_{j\in\eN}$ is nonincreasing and (when the sequence is infinite) $\ell_j\to 0$ as $j\to\infty$.
	(In the sequel, we will not always assume that $(\ell_j)_{j\in\eN}$ has been written in nonincreasing order.)
	We note that any choice of open set $\O\stq\eR$ satisfying the above properties is called a {\em geometric realization} of $\mathcal L$.
	
	Conversely, given a bounded fractal string $(\ell_j)_{j\in\eN}$, there are many different ways to associate to it an open set $\O$ of finite length and such that $|\O|_1=\sum_{j\geq 1}\ell_j$.
	There is, however, a canonical way to do so; see \cite[pp.\ 88--89]{fzf}.
	
	We close this remark by recalling that if $\mathcal L$ is an infinite sequence of positive numbers, then $D(\zeta_{\mathcal L})$ coincides with the (upper) Minkowski dimension of $\mathcal L$ (i.e., of $\pa\O$, for any choice of geometric realization of $\mathcal L$, in the above sense; see 
	\cite[Theorem 1.10]{lapidusfrank12}). 
\end{remark}

\medskip

\begin{example}[{\it The Cantor string}]
	\label{ex:1.2}
	A typical example of a bounded fractal string is the {\em Cantor string};
	see \cite[Chapter 1]{lapidusfrank12}.
	In the sense of Remark \ref{rem:1.1} just above, it is associated with the bounded open subset of $\eR$ given by $\O_{CS}:=[0,1]\setminus C$, the complement in $[0,1]$ of the ternary Cantor set $C$, and consists of the ``middle-thirds'' (that is, of all the deleted intervals in the standard construction of the Cantor set $C$).
	Then, $\mathcal{L}={\mathcal L}_{CS}:=(\ell_j)_{j\geq 1}$ consists of the following infinite sequence 
	\begin{equation}
		\label{eq:3.1/4}
		\frac{1}{3},\frac{1}{9},\frac{1}{9},\frac{1}{27},\frac{1}{27},\frac{1}{27},\frac{1}{27},\ldots,
	\end{equation}
	where $3^{-j}$ appears with the multiplicity $2^{j-1}$ (for $j=1,2,\ldots$).
	
	Observe that the boundary of the Cantor string is the classic Cantor set $C$: $\partial\O_{CS}=C$.
	
	Finally, a simple computation (based on \eqref{eq:3} and \eqref{eq:3.1/4} and followed by an application of the principle of analytic continuation), shows that $\zeta_{\mathcal{L}}=\zeta_{\mathcal{L}}(s)$ (also denoted by $\zeta_{CS}(s)$) is meromorphic in all of $\Ce$ and is given by
	\begin{equation}
		\label{eq:3.1/2}
		\zeta_{\mathcal{L}}(s)=\frac{1}{3^s-2}=\frac{3^{-s}}{1-2\cdot 3^{-s}},
	\end{equation}
	for all $s\in\Ce$; see \cite[Subsection 1.2.2]{lapidusfrank12}.
\end{example}

\section{Paramorphic functions and their paramorphic continuations}

It has been noticed that there are (nontrivial) bounded fractal strings without any complex dimensions in the classical sense (viewed as poles of a meromorphic extension of the associated geometric zeta function).
As an example, see the fractal string $\mathcal{L}_{\infty}$ constructed in \cite[Example 3.3.7 on p.\ 215]{fzf} or in \cite{ra1}.
In this case, the geometric zeta function $\zeta_{\mathcal{L}}$ does not have any poles but has essential singularities.
Therefore, there is a natural need to extend the notion of complex dimensions, in order to include essential singularities as well.
To achieve this, we need a more general definition of an extension (of a geometric zeta function)
than just a meromorphic extension to an open right half-plane (or some more general domain) of the complex plane.
This leads in a natural way to the notions of paramorphic extensions and paramorphic functions.
An additional justification is provided by the fact that singularities which are not poles (of a fractal zeta function) also have a natural geometric meaning in our context because like the poles, they often contribute to the corresponding fractal tube formula; see \cite{log}.


\begin{defn}\label{para}
Let $U$ be a nonempty connected open subset of the complex plane, and let $S:=\{s_k:k\in J\}$ be a subset (possibly empty) of isolated points of $U$.\footnote{In particular, the set $S$ does not have an accumulation point in $U$. Here, the set $J$ denotes an an arbitrary index set.}
Let $f:U\setminus S\to\Ce$ be a holomorphic function. Then, we say in short that the function $f$ is {\em paramorphic} in~$U$.
\end{defn}

\begin{remark}\label{apriori}
The set $S$ appearing in Definition \ref{para} is clearly at most countable, and the set of (possible) accumulation points of $S$ is contained in the topological boundary $\pa U$ of $U$. Indeed, since we assume the function $f:U\setminus S\to\Ce$ to be holomorphic, then the set $U\setminus S$ must a priori be open. In other words, the set $S$ is closed with respect to the relative topology of $U$.
\end{remark}

Obviously, all meromorphic functions are automatically paramorphic but the converse is, of course, not true.

\begin{example}\label{paramex}
	The function $f(z)=\E^{1/(z-z_0)}$ is paramorphic in $\mathbb{C}$. Here, $z_0\in\Ce$ is the only singularity of $f$, and it is essential.
\end{example}

\begin{lemma}\label{Sclosed}
Assume that a complex-valued function $f$ is paramorphic $($in the sense of Definition \ref{para}$)$ on a nonempty connected open subset $U$ of the complex plane. Then the set $S=S(f)$ of its nonremovable isolated singularities contained in $U$ $($i.e., the set of poles and essential singularities of $f$ contained in $U$$)$ is closed with respect to the relative topology of $U$.
\end{lemma}

\begin{proof}
Assume, contrary to the claim, that the set $S$ is not closed. Then there exists $s_0\in(\mbox{\rm Cl}\,S\setminus S)\cap U$. On the one hand, $f$ is holomorphic at $s_0$, since $s_0\in U\setminus S$. On the other hand, there is a sequence $(s_k)_{k\ge1}$ of nonremovable singularities of $f$ converging to $s_0$ as $k\to\ty$, which is impossible. This proves the lemma.
\end{proof}

As we can see, by saying that a complex-valued function $f$ is paramorphic on $U$, we mean that $f:U\to\Ce$ is differentiable (i.e.,
holomorphic) at all points of $U$ except on a subset $S$ of isolated singularities of $f$. Each $s_0\in S$ is either a removable singularity, or a pole, or an isolated essential singularity. If we exclude removable singularities from the set $S$, then $S$ is uniquely determined by $f$, consisting of its poles and isolated essential singularities contained in $U$.

For a fixed nonempty connected open subset $U$ of the complex plane, the vector space of all functions paramorphic on $U$ is denoted by ${\rm Par}(U)$.

\begin{remark}
	We point out that the notion of a paramorphic function is closely related to the class $\mathcal{S}$ of functions introduced by A.\ Bolsch in \cite{Bol1,Bol2} (see also the class $\mathcal{K}$ from \cite{Bak,Dom}) for studying iterations of complex maps (from the dynamical perspective) which are meromorphic except in a ``small'' set.
	Namely, a function $f:\overline{\Ce}\to\overline{\Ce}$ is said to be in the class $\mathcal{S}$ if there exists a closed countable set $A(f)\subseteq\overline{\Ce}$ such that $f$ is meromorphic in $\overline{\Ce}\setminus A(f)$ but in no proper superset.\footnote{Here, $\overline{\Ce}$ denotes the Riemann sphere, i.e., the one-point compactification of $\Ce$: $\overline{\Ce}:=\Ce\cup\{\infty\}$.}
	
	The above definition is more general than the definition of a paramorphic function since the set $A(f)$ may also contain non-isolated singularities that arise as accumulation points of isolated singularities of $f$.
	On the other hand, a paramorphic function $f:U\to\Ce$ cannot have any non-isolated singularities in the open domain $U\subseteq\Ce$.
	In the general theory of complex dimensions, we conjecture that only isolated singularities of fractal zeta functions should be considered as ``proper'' complex dimensions of the associated fractal set.
	A strong indication of this is the fact that they have a direct geometric meaning since these complex dimensions appear as co-exponents in the asymptotics of the fractal tube formula of the given set, whereas the non-isolated singularities are a kind of a byproduct of the isolated ones, i.e., of the ``proper'' complex dimensions.
\end{remark}


\begin{defn}\label{parae}
Assume that $U$ and $V$ are connected open subsets of the complex plane, and $f\in{\rm Par}(U)$, $g\in{\rm Par}(V)$.
If $U\stq V$ and $g|_U=f$ (except for the set of isolated singularities of $f$), we say that $g$ is a {\em paramorphic extension} of $f$.
\end{defn}

\begin{remark}
In Definition \ref{parae}, by writing $g|_U=f$, we mean that in fact $g|_{U\setminus\, S}=f$, where $S=S(f)$ is the set of isolated singularities of $f$.
As in Remark~\ref{apriori} and Lemma~\ref{Sclosed}, the set $S(f)$ is closed in the relative topology of $U$, since $f|_{U\setminus\, S(f)}$ is holomorphic.
\end{remark}

\begin{remark}
If by ${\rm Hol}(U)$ and ${\rm Mer}(U)$ we denote the vector spaces of functions which are, respectively, holomorphic and meromorphic  on a nonempty connected open subset $U$ of $\Ce$, then ${\rm Hol}(U)\stq{\rm Mer}(U)\stq{\rm Par}(U)$.
\end{remark}

The following result shows that a paramorphic extension $g\in{\rm Par}(V)$ of $f\in{\rm Par}(U)$ in Definition \ref{parae}
is uniquely determined by $f$.

\begin{theorem}[Unique paramorphic continuation principle]\label{cp} Let $U$ and $V$ be non\-empty connected open subsets of the complex plane $\Ce$ and $U\stq V$. If $g_1,g_2\in{\rm Par}(V)$ and $g_1|_U=g_2|_U$, then $g_1=g_2$. In other words, the sets of nonremovable isolated singularities of $g_1$ and $g_2$ coincide, and $g_1=g_2$ on the complement of their common set of singularities in~$V$. 
\end{theorem}

\begin{proof}
Let $S=S(g_1)$ be the set of nonremovable isolated singularities of $g_1$.
Then, according to Definition~\ref{para}
(and Remark~\ref{apriori} along with Lemma~\ref{Sclosed}), $U\setminus S$ is an open set and $g_1$ is holomorpohic in all of
$V\setminus S$. Therefore, since $g_2$ coincides with $g_1$ on $U\setminus S$, and since $V\setminus S$ is
open and connected, it follows from the principle of analytic
continuation that $g_2$ coincides with $g_1$ on all of $V\setminus S$. As a result,
$g_2$ is holomorphic in all of $V\setminus S$ and hence, $S(g_2)$ is contained in
$S(g_1) = S$.

Now, by the symmetry of the hypotheses on $g_1$ and $g_2$ (in the
statement of Theorem \ref{cp}), we could apply the same reasoning by
interchanging the roles of $g_1$ and $g_2$ and conclude that  $S(g_1)$ is
also contained in $S(g_2)$. Hence, $g_1$ and $g_2$ have a common set of
nonremovable singularities $S$, and $g_1$ and $g_2$ coincide on $V\setminus S$; that is, $g_1 =
g_2$.
\end{proof}

We also provide the following result, which shows that the set ${\rm Par}(U)$ of paramorphic functions on a given connected open subset $U\stq\Ce$ is closed under multiplications; i.e., it is an algebra.

\begin{theorem}
Let $U$ be a given connected open subset of the complex plane. Then, the set of paramorphic functions ${\mathrm Par}(U)$ is a unital algebra $($with respect to pointwise multiplication$)$.
\end{theorem}

\begin{proof}
The unit element in this algebra is, of course, the function $1\in{\rm Par}(U)$ defined by $1(s)=1$ for all $s\in U$.
For $f_1,f_2\in{\rm Par}(U)$, it is easy to see that also $f_1\cdot f_2\in{\rm Par}(U)$. Namely, if $f_j:U\setminus S_j\to\Ce$, $j=1,2$, are two holomorphic functions, where $S_j$ are the corresponding sets of isolated singularities of $f_j$, for $j=1,2$, then the product $f_1\cdot f_2$ is well defined and holomorphic on $U\setminus(S_1\cup S_2)$. (Here, some elements of $S_1\cup S_2$ may be removable singularities of $f_1\cdot f_2$, due to possible cancellations.) Hence, according to Definition~\ref{para}, the product $f_1\cdot f_2$ is paramorphic on $U$.
\end{proof}

In the following definition, we introduce the notion of the `abscissa of paramorphic continuation' of a given paramorphic function, which is analogous to that of the `abscissa of meromorphic continutation' of a given meromorphic function.

\begin{defn}\label{D_par}
Let $\a$ be a real number and let $\{\re s>\a\}$ be the corresponding open right half-plane in $\Ce$. Assume that $f:\{\re s>\a\}\to\Ce$ is a Dirichlet-type function (or, in short, DTI; see, e.g., \cite{fzf}, esp., Appendix A), such that $f$ is paramorphic on $\{\re s>\a\}$, for some $\a\in\eR$.\footnote{Otherwise, we let $D_{\rm par}(f)=+\infty$, which means that $f$ cannot be paramorphically extended to any (nonempty) right half-plane.} The {\em abscissa of paramorphic continuation $D_{\rm par}(f)$} of $f$ is defined as the infimum of all real numbers $\beta$, with $\beta\le\a$, such that $f$ can be paramorphically extended from $\{\re s>\a\}$ to $\{\re s>\b\}$.\footnote{We also allow for $D_{\rm par}(f)=-\infty$, which means that $f$ can be paramorphically extended to all of $\Ce$.} Equivalently, 
$\{\re s>D_{\rm par}(f)\}$ is the largest open right half-plane, to which $f$ can be paramorphically extended. (It is easy to deduce from Theorem~\ref{cp} that this notion is well defined.)\footnote{Indeed, if $f$ is paramorphic on each element of a family of right half-planes, $\{\re s>\alpha_i\}_{i\in I}$, then (by Theorem \ref{cp}) it is paramorphic on the union of these right-half planes, namely, on the right-half plane $\{\re s>\alpha\}$, where $\alpha:=\inf_{i\in I}\alpha_i$.} Clearly, 
$$
-\ty\le D_{\rm par}(f)\le D_{\rm mer}(f)\le D(f)\leq +\infty.
$$
Furtermore, given $D_{\ty}\in\eR$, the vertical line $\{\re s=D_\ty\}$ is said to be a {\em paramorphic barrier} of $f$ if $f$ cannot be paramorphically continued to a connected open set $V$ containing the open right half-plane $\{\re s>D_\ty\}$ as a proper subset.
\end{defn}

If $f$ is a DTI of the form of a geometric zeta function, i.e., $f=\zeta_{\mathcal L}$ for some bounded fractal string $\mathcal L$, then clearly, $0\le D_{\rm par}(f)\le D_{\rm mer}(f)\le D(f)\le 1$. It is also clear that the notion of a paramorphic barrier, introduced in Definition~\ref{D_par} above, can be extended to a much more general setting.
\medskip

We are now ready to state the main result of this paper.

\begin{theorem}\label{main}
Let $D_{\ty}$, $D_1$ and $D$ be three prescribed real numbers belonging to the interval $[0,1]$ and such that $D_\ty< D_1\le D$. Then, there exists an explicit $($i.e., explicitly constructible$)$  bounded fractal string $\mathcal L$ such that the corresponding geometric zeta function $\zeta_{\mathcal L}$ can be paramorphically extended to the open right half-plane $\{\re s>D_{\ty}\}$ and
\begin{equation}
D_{\rm par}(\zeta_{\mathcal L})=D_\ty,\q
D_{\rm mer}(\zeta_{\mathcal L})=D_1,\q
D(\zeta_{\mathcal L})=\dim{\mathcal L}=D.
\end{equation}
In addition to this, it can be achieved that the line $\{\re s=D_\ty\}$ coincides with the paramorphic barrier of $\zeta_{\mathcal L}$ $($in the sense of Definition \ref{D_par} above$)$, while the vertical open strip $\{D_\ty<\re s<D_1\}$ contains infinitely many essential singularities of $\zeta_{\mathcal L}$, and such that the paramorphic barrier coincides with the set of accumulation points of the set of essential singularities of $\zeta_{\mathcal L}$.
\end{theorem}
\medskip

We postpone the proof of Theorem~\ref{main} until Section~\ref{mains} (more precisely, until Subsection~\ref{mainss}).

For general references on complex analysis and the singularities of complex-valued functions, we refer, e.g., to \cite{Ahl}, \cite{Conw1,Conw2}, \cite{Ebe}, \cite{Schl} and \cite[esp., Subsection 1.3.2.]{fzf}.

\section{Generalized Cantor strings of finite and infinite orders and their geometric zeta functions}\label{sec:3}

\subsection{Generalized Cantor strings of finite order} Let $r_j$, with $j=1,\dots,m$, be positive real numbers such that $r_1+\dots +r_m<1$.
Let $\mathcal L(r_1,\dots,r_m)$ be the self-similar fractal string defined as the nonincreasing sequence of all monomial terms of the form $r_1^{\a_1}\dots r_m^{\a_m}$, with $(\a_1,\dots,\a_m)\in(\eN\cup \{0\})^m$.
It can be shown (see \cite[Chapters 2 and 3]{lapidusfrank12}) that the corresponding geometric zeta function is given by
\begin{equation}\label{zetaLr}
\zeta_{\mathcal L(r_1,\dots,r_m)}(s)=\frac1{1-\sum_{j=1}^mr_j^s},
\end{equation}
for all $s\in\Ce$.
This is established by first verifying Eq.~\eqref{zetaLr} via a direct computation, valid for all $s\in\Ce$ with $\re s$ sufficiently large,\footnote{Namely, for all $s\in\Ce$ with $\re s>D_{\mathcal{L}}$, where $D_{\mathcal{L}}$ is the Minkowski (or box) dimension of $\mathcal{L}=\mathcal{L}(r_1,\ldots,r_m)$, which, in the present case, coincides with the {\em similarity dimension} of $\mathcal L$, i.e., the unique {\em real} solution of the Moran equation \cite{Mor} (see also, e.g., \cite{Fal}) $\sum_{j=1}^mr_j^s=1$.} and then upon meromorphic continuation, by deducing that \eqref{zetaLr} holds, in fact, for all $s\in\Ce$.

For example, by choosing $m=2$ and $r_1=r_2=1/3$, we obtain the {\em Cantor string} 
$$
\mathcal L(1/3,1/3)=(\ell_j)_{j\in\eN}.
$$ 
It corresponds to the nonincreasing sequence of lengths of deleted open intervals obtained during the construction of the usual Cantor's ternary set $C^{(2,1/3)}$ scaled by the factor $3$, i.e., starting with the interval $[0,3]$ instead of $[0,1]$; see \cite[{{\em ibid}}]{lapidusfrank12} or \cite[Definition 3.3.1 and Theorem 3.3.3]{fzf}. Furthermore, in light of Eq.~\eqref{zetaLr} and in keeping with the above explanations, we see that
$\zeta_{\mathcal L(1/3,1/3)}(s)=1/(1-2\cdot 3^{-s})$ for all $s\in\Ce$ such that $\re s>\log_32$. 
As was explained above in the case of a general self-similar string, the geometric zeta function $\zeta_{\mathcal L(1/3,1/3)}$ can then be  meromorphically extended to the whole complex plane by letting $\zeta_{\mathcal L(1/3,1/3)}(s)=1/(1-2\cdot 3^{-s})$ for all $s\in\Ce$.

Let $m$ be a positive integer such that $m\ge2$, and let $a\in(0,1/m)$. Let us define the {\em generalized Cantor string} 
\begin{equation}\label{Lmas}
{\mathcal L}^{(m,a)}=\mathcal L(\underbrace{a,\dots,a}_{\mbox{\scriptsize$m$ times}}).
\end{equation} 
Here, by using Eq.~\eqref{zetaLr}, we obtain that
\begin{equation}\label{Lmer}
\zeta_{{\mathcal L}^{(m,a)}}(s)=\frac1{1-\sum_{j=1}^ma^s}=\frac1{1-m\cdot a^s},
\end{equation}
for all $s\in\Ce$ with $\re s>\log_{1/a}m$. This geometric zeta function can then be meromorphically extended to the whole complex plane, so that \eqref{Lmer} holds for all~$s\in\Ce$.

For any fixed integer $n\ge1$, we introduce the {\em generalized Cantor string of $n$-th order}, ${\mathcal L}^{(m,a)}_n$, defined inductively by
\begin{equation}
{\mathcal L}^{(m,a)}_1:={\mathcal L}^{(m,a)}\q\mbox{and}\q
{\mathcal L}^{(m,a)}_n:={\mathcal L}^{(m,a)}_{n-1}\otimes {\mathcal L}^{(m,a)}\q\mbox{for\,\, $n\ge2$}.
\end{equation}
In other words, we iterate multiplying ${\mathcal L}^{(m,a)}$ by itself, using the tensor product of fractal strings; that is, for every integer $n\geq 1$, 
\begin{equation}
{\mathcal L}^{(m,a)}_n:=\bigotimes_{j=1}^n {\mathcal L}^{(m,a)}.
\end{equation}
The geometric zeta function of ${\mathcal L}^{(m,a)}_n$ can be explicitly computed (initially, for all $s\in\Ce$ with $\re s$ large enough) and then meromorphically extended to the whole complex plane. We first have
\begin{equation}
\zeta_{{\mathcal L}^{(m,a)}_2}(s)=\zeta_{{\mathcal L}^{(m,a)}_1}(s)\cdot \zeta_{{\mathcal L}^{(m,a)}}(s)=\frac1{1-m\cdot a^s}\cdot\frac1{1-m\cdot a^s}=\frac1{(1-m\cdot a^s)^2}.
\end{equation}
and then by induction, for each $n\ge1$ and all $s\in\Ce$,
\begin{equation}
\zeta_{{\mathcal L}^{(m,a)}_n}(s)=\frac1{(1-m\cdot a^s)^n}.
\end{equation}
Here, we have used the multiplicative property of the geometric zeta function with respect to the tensor products of fractal strings; see \cite[Lemma~3.3.2]{fzf}. The total length of the generalized Cantor string of $n$-th order ${\mathcal L}^{(m,a)}_n$ is given by
\begin{equation}
|{\mathcal L}^{(m,a)}_n|_1=\zeta_{{\mathcal L}^{(m,a)}_n}(1)=\frac1{(1-m\cdot a)^n}.
\end{equation}
Note that $|{\mathcal L}^{(m,a)}_n|_1\to+\ty$ as $n\to\ty$, exponentially fast  as a function of~$n$.

The set of {\em complex dimensions} of the fractal string ${\mathcal L}^{(m,a)}_n$, denoted by $\dim_\Ce{\mathcal L}^{(m,a)}_n$, is defined as the set of poles (in $\Ce$) of the associated geometric zeta function $\zeta_{{\mathcal L}^{(m,a)}_n}$. In this case, the poles of ${\mathcal L}^{(m,a)}_n$ are all of order $n$ (i.e., the complex dimensions of ${\mathcal L}^{(m,a)}_n$ are of multiplicity $n$), and they form an arithmetic sequence contained in the vertical line $\{\re s=\log_{1/a}m\}$ of the complex plane:
\begin{equation}
\dim_\Ce{\mathcal L}^{(m,a)}_n=\log_{1/a}m+\frac{2\pi}{\log (1/a)}\I\Ze.
\end{equation}
The above construction of the fractal string ${\mathcal L}^{(m,a)}_n$, as well as the computation of its geometric zeta function, are a natural extension of the ones provided in \cite[Example 3.3.7 on p.\ 215]{fzf} in the case when $m=2$ and $r_1=r_2=1/3$. For the general theory of the complex dimensions of fractal strings, see \cite{lapidusfrank12} and \cite{fzf}.

It is easy to explicitly compute the coefficients $c^j_l$, with $j\ge1$, appearing in the Laurent expansion
\begin{equation}
\zeta_{{\mathcal L}^{(m,a)}_n}(s)=\sum_{l=-n}^{\ty}c^j_l(s-D_j)^l
\end{equation}
of the geometric zeta function $\zeta_{{\mathcal L}^{(m,a)}_n}$ near any of the poles $s_j:=\log_{(1/a)}m+\frac{2\pi}{\log(1/a)}\I j$ of $\zeta_{{\mathcal L}^{(m,a)}_n}$, for a fixed value of $j\in\Ze$ and for a prescribed integer~$n\ge1$.
For example, we have that
\begin{equation}
\begin{aligned}
c^j_{-n}&:=\lim_{s\to s_j}{(s-s_j)^n}\zeta_{{\mathcal L}^{(m,a)}_n}(s)=
\Big(\lim_{s\to s_j}\frac{s-s_j}{1-m\cdot a^s}\Big)^n\\
&=\Big(\frac1{a^{s_j}\log(1/a)}\Big)^n=\Big(\frac m{\log(1/a)}\Big)^n.
\end{aligned}
\end{equation}
It is interesting to note that the value of $c^j_{-n}$ is, in fact, independent of $j\in\Ze$.

Other coefficients of the form $c_l=c_{-n+r}$, with $r\ge1$, can be easily computed as well, since $c_{-n+r}=\lim_{s\to s_j}\frac{\dd^r}{\dd s^{r}}\Big.\Big|_{s=s_j}[{(s-s_j)^n}\zeta_{{\mathcal L}^{(m,a)}_n}(s)]$.

\subsection{Generalized Cantor strings of infinite order}
Now, we can define the {\em generalized Cantor string of infinite order} as the following infinite disjoint union of scaled generalized Cantor strings of finite orders:
\begin{equation}\label{Lty}
{\mathcal L}^{(m,a)}_\ty:=\bigsqcup_{n=1}^\ty(n!)^{-1}{\mathcal L}^{(m,a)}_n.
\end{equation}
Its geometric zeta function is then given by
\begin{equation}\label{zetaLmaty}
\begin{aligned}
\zeta_{{\mathcal L}^{(m,a)}_\ty}(s)&=\sum_{n=1}^\ty\zeta_{(n!)^{-1}{\mathcal L}^{(m,a)}_n}(s)\\
&=\sum_{n=1}^\ty(n!)^{-s}\zeta_{{\mathcal L}^{(m,a)}_n}(s)=
\sum_{n=1}^{\ty}\frac{(1-m\cdot a^s)^{-n}}{(n!)^{s}}.
\end{aligned}
\end{equation}
Using the Weierstrass $M$-test, it is easy to see that ${\mathcal L}^{(m,a)}_\ty$ can be paramorphically extended to the open right half-plane $\{\re s>0\}$; that is, $\zeta_{{\mathcal L}^{(m,a)}_\ty}\in{\rm Par}(\{\re s>0\})$. Here, the set
$\log_{1/a}m+\frac{2\pi}{\log (1/a)}\I\Ze$ consists of essential singularities of the geometric zeta function $\zeta_{{\mathcal L}^{(m,a)}_\ty}$, and there are no other isolated singularities. (For $m=2$ and $r=1/2$, this construction has been described in \cite[Example 3.3.7]{fzf}; see also \cite{ra1} and \cite{ra2}.) In light of Eq.~\eqref{zetaLmaty}, we see that the total length of the string ${\mathcal L}^{(m,a)}_\ty$ is given by
\begin{equation}\label{lengthLty}
|{\mathcal L}^{(m,a)}_\ty|_1=\zeta_{{\mathcal L}^{(m,a)}_\ty}(1)=\sum_{n=1}^{\ty}\frac{(1-m\cdot a)^{-n}}{n!}=\exp\Big(\frac1{1-m\cdot a}\Big)-1.
\end{equation}
In particular, $\mathcal{L}^{(m,a)}_\ty$ is a bounded fractal string and we always have that $|{\mathcal L}^{(m,a)}_\ty|_1>\E-1>0$.

\begin{remark} 
We do not know whether $\zeta_{{\mathcal L}^{(m,a)}_\ty}$ can be paramorphically extended to an open right half-plane $\{\re s>\beta\}$, for some $\beta<0$. 
\end{remark}



\subsection{Power series of bounded fractal strings}

Let $X$ be the set of all bounded fractal strings.  In Subsection \ref{notation}, we have introduced two binary operations, which can be viewed as the operations of addition and multiplication on $X$, defined as the disjoint union $\sqcup$ of fractal strings and the tensor product $\otimes$, respectively. It is easy to check that $({\mathcal L}_1\sqcup{\mathcal L}_2)\otimes\mathcal L_3=({\mathcal L}_1\otimes{\mathcal L}_3)\sqcup({\mathcal L}_2\otimes{\mathcal L}_3)$, for any $\mathcal L_n\in X$, $n=1,2,3$. In this manner, we have obtained a commutative unital semiring $(X,\sqcup,\otimes)$ (without the zero element).\footnote{If zero in $X$ were defined as the one element sequence $(0)$, then $X$ should contain  $(0)\otimes\mathcal{L}=(0,0,,\ldots)$, which is an infinite sequence of zeros. This means that this string has the real number 0 with {\em infinite} multiplicity which we cannot permit. Otherwise, the disjoint union of a nonzero string $\mathcal{L}$ and $0$ in $X$ is not well defined (i.e, it cannot be ordered as a nonincreasing sequence of reals).} The unit element in this semiring is $\mathcal E:=(1)$. This structure is not a ring, since the elements of $X$ do not possess additive inverses with respect to the binary operation~$\sqcup$.

We also have the operation of scalar multiplication of bounded fractal strings $\mathcal L:=(\ell_j)_{j\ge1}$ with positive real numbers $c$, where the resulting fractal string is $c\mathcal L:=(c\ell_j)_{j\ge1}$. The set $X$, viewed with respect to $\sqcup$ as addition and with respect to scalar multiplication, is clearly a positive convex cone, since for any positive real numbers $c$ and $d$ and any two fractal strings $\mathcal L_1,\mathcal L_2\in X$, we have that $c\mathcal L_1\sqcup d\mathcal L_2\in X$.

We are now ready to introduce the notion of a {\em power series of bounded fractal strings} in $X$, as follows. Let $F(z):=\sum_{n=0}^{\ty}c_nz^n$ be the usual power series of complex numbers $z$, where we assume that the coefficients $c_n$ are nonnegative real numbers for all integers $n\ge0$ and $c_n>0$ for at least one $n\ge0$, such that the radius of convergence $R$ of the series $F$ is positive (or infinite). For any fixed fractal string $\mathcal L:=(\ell_j)_{j\ge1}\in X$ such that $|\mathcal L|_1:=\sum_{j\ge1}\ell_j<R$ (i.e., of {\em total length} less than $R$), we can define the corresponding bounded fractal string $F(\mathcal L)$ by
\begin{equation}
	F(\mathcal{L}):=\bigsqcup_{n=0}^{\ty}c_n{\mathcal L}^n,
\end{equation}
where ${\mathcal L}^n$ is the tensor product of $n$ copies of $\mathcal L$ for $n\ge1$, while ${\mathcal L}^0:=\mathcal E$. It is easy to verify that the fractal string $F(\mathcal L)$ is bounded: $|F(\mathcal L)|_1=\sum_{n=1}^{\ty}c_n|\mathcal L|_1^n<\ty$, that is, $F(\mathcal L)\in X$. In this way, we have obtained the mapping
$$
F:\{\mathcal L\in X: |\mathcal L|_1<R\}\to X.
$$
In particular, if $R=+\ty$, we have the mapping $F:X\to X$.

As an example, if we consider the function $F(z):=\exp (z)$, then $c_n=(n!)^{-1}$ for all $n\ge0$ and $R=+\ty$. We see that for any bounded fractal string $\mathcal L\in X$, the {\em exponential fractal string of $\mathcal L$}, that is,
\begin{equation}
	\exp(\mathcal L)=\bigsqcup_{n=0}^{\ty}(n!)^{-1}{\mathcal L}^n,
\end{equation}
is well defined, i.e., it belongs to $X$. Hence,
\begin{equation}
	\zeta_{\exp(\mathcal L)}(s)=\sum_{n=0}^{\ty}(n!)^{-s}\zeta_{\mathcal L}(s)^n, 
\end{equation}
for all $s$ in the open right half-plane $\{\re s>D(\zeta_{\mathcal L})\}$.

In particular, if we take $\mathcal L={\mathcal L}^{(m,a)}$ (the generalized Cantor string defined in Eq.~\eqref{Lmas}) and if ${\mathcal L}^{(m,a)}_{\ty}$ is the generalized Cantor string of infinite order (introduced in Eq.~\eqref{Lty}),  then
\begin{equation}
	\exp({\mathcal L}^{(m,a)})={\mathcal L}_\ty^{(m,a)}\sqcup\{\mathcal E\}
\end{equation}
and
\begin{equation}\label{Lma1}
	\zeta_{\exp({\mathcal L}^{(m,a)})}(s)=\zeta_{{\mathcal L}_\ty^{(m,a)}}(s)+1,
\end{equation}
for all complex numbers $s$ in the open right half-plane $\{\re s>\log_{1/a}m\}$.
In other words, the geometric zeta functions of fractal strings $\exp({\mathcal L}^{(m,a)})$ and ${\mathcal L}_\ty^{(m,a)}$ coincide up to the additive constant~$1$.
Of course, if we take $G(z):=\exp(z)-1=\sum_{n\ge1}\frac{z^n}{n!}$, then we precisely have equality in the counterpart of Eq.\ \eqref{Lma1}; i.e., $\zeta_{G({\mathcal L}^{(m,a)})}(s)=\zeta_{{\mathcal L}_\ty^{(m,a)}}(s)$.

In a similar way, for any fractal string $\mathcal L\in X$ of total length less than $1$, we can define the fractal string $F(\mathcal L)=(1-\mathcal L)^{-1}$, generated by the power series $F(z):=(1-z)^{-1}=\sum_{n\ge0}z^n$, as well as $G(\mathcal L)=-\log(1-\mathcal L)$, generated by the function $G(z):=-\log(1-z)=\sum_{n\ge1}\frac{z^n}n$. For any $\mathcal L\in X$, we can analogously define the bounded fractal strings $\cosh\mathcal L$ (generated by $F(z):=\sum_{n\ge0}\frac{z^{2n}}{(2n)!}$) and $\sinh\mathcal L$ (generated by $F(z):=\sum_{n\ge0}\frac{z^{2n+1}}{(2n+1)!}$), etc.
\medskip

The following result connects the geometric zeta functions of the fractal strings $F(\mathcal L)$ and $\mathcal L$.

\begin{prop}\label{FL}
Let $\mathcal L\in X$, and let $F(z)=\sum_{n=0}^\ty c_nz^n$ be a power series with nonnegative coefficients, where $c_n>0$ for at least one $n\ge0$ and with radius of convergence $R>0$.

Then, for any fractal string $\mathcal L\in X$ of total length less than $R$ $($i.e., $|\mathcal L|_1<R$$)$, we have that
\begin{equation}\label{FL0}
\zeta_{F(\mathcal L)}(s)=\sum_{n=0}^\ty c_n^s\zeta_{{\mathcal L}}(s)^n,
\end{equation}
for all complex numbers $s$ in the open right half-plane $\{\re s>D(\zeta_{\mathcal L})\}$, where $D(\zeta_{\mathcal L})$ is the abscissa of absolute convergence of $\zeta_{\mathcal L}$ $($i.e., the Minkowski dimension of $\mathcal L$ if $\mathcal L$ is an infinite sequence$)$. In particular, if $D(\zeta_{\mathcal L})<1$, then
\begin{equation}
\zeta_{F(\mathcal L)}(1)=F(|\mathcal L|_1).
\end{equation}
%
\end{prop}

\begin{proof}
For any $s$ in $\{\re s>D(\zeta_{\mathcal L})\}$, we have that

\begin{equation}\label{FLc}
\zeta_{F(\mathcal L)}(s)=\sum_{n=0}^\ty \zeta_{c_n{\mathcal L}^n}(s)=\sum_{n=0}^\ty c_n^s\zeta_{{\mathcal L}^n}(s)
=\sum_{n=0}^\ty c_n^s\zeta_{{\mathcal L}}(s)^n,
\end{equation}
where we have used the fact that for any three fractal strings $\mathcal L_1$, $\mathcal L_2$ and $\mathcal L$ in $X$ and for any positive real number $c$, we have that $\zeta_{\mathcal L_1\sqcup\mathcal L_2}(s)=\zeta_{\mathcal L_1}(s)+\zeta_{\mathcal L_2}(s)$, $\zeta_{c\mathcal L}(s)=c^s\zeta_{\mathcal L}(s)$ and $\zeta_{\mathcal L_1\otimes\mathcal L_2}(s)=\zeta_{\mathcal L_1}(s)\cdot\zeta_{\mathcal L_2}(s)$ (in particular, by mathematical induction, we have that $\zeta_{\mathcal L^n}(s)=\zeta_{\mathcal L}(s)^n$, for any $n\ge 2$).

If $D(\zeta_{\mathcal L})<1$, then \eqref{FLc} implies that
\begin{equation}
\zeta_{F(\mathcal L)}(1)=\sum_{n=0}^\ty c_n\zeta_{{\mathcal L}}(1)^n=F(\zeta_{\mathcal L}(1))=F(|\mathcal L|_1).
\end{equation}
%
%
\end{proof}


\section{Geometric zeta functions with prescribed abscissa of paramorphic continuation}\label{mains}

This section is divided into two subsections: in Subsection \ref{mainss}, we establish one of the key results of this paper (Theorem \ref{main}), which establishes the existence of suitable paramorphic (and complex-valued) fractal zeta functions with prescribed abscissae of paramorphic, meromorphic and absolute convergence, respectively. Moreover, in Subsection \ref{harm}, based in part on this result, we construct suitable (real-valued) harmonic functions that are associated with paramorphic geometric zeta functions and have interesting sets of essential singularities.

\subsection{Construction of a class of paramorphic fractal zeta functions via a sequence of generalized Cantor strings}\label{mainss}
In this subsection, using the results of Section \ref{sec:3}, we construct a fractal string $\mathcal L$ such that the corresponding geometric zeta function $\zeta_{\mathcal L}$ has prescribed values of abscissae of paramorphic continuation $D_{\rm par}(\zeta_{\mathcal L})$ (see Definition \ref{D_par}), of meromorphic continuation $D(\zeta_{\rm mer})$ and of absolute convergence $D(\zeta_{\mathcal L})$. The construction of $\mathcal L$ is based on a careful choice of a suitable sequence of generalized Cantor strings. The corresponding precise result is stated in Theorem~\ref{main}, to which we refer the reader and which we now establish.
\medskip

\begin{proof}[Proof of Theorem~\ref{main}]
{\em Case $(i)$}: We first consider the case when $D_\ty< D_1=D$. 
As we have seen, each fractal string ${\mathcal L}^{(m,a)}_\ty$ is bounded for any integer $m\ge2$ and for any real number $a\in(0,1/m)$; see Eq.~\eqref{lengthLty} above. Let $(D_k)_{k\ge2}$ be any decreasing sequence of real numbers converging to $D_{\ty}$ as $k\to\ty$ and such that $D_2<D_1$.

Let $(m_k)_{k\ge1}$ be a strictly increasing sequence of integers diverging to $+\ty$ as $k\to\ty$, such that $m_1\ge 2$. Next, we define positive real numbers $a_k$ by the following equality: $D_k=\log_{1/a_k}m_k$; that is, $a_k:=m_k^{-1/D_k}$, for all $k\ge1$. We have that $m_ka_k=m_k^{1-1/D_k}<1$; i.e., $a_k\in(0,1/m_k)$, for all~$k\ge1$. 

Now, we introduce the following sequence of bounded fractal strings:
\begin{equation}\label{Lkstring}
{\mathcal L}_k:=\frac{2^{-k}}{L_k}{\mathcal L}^{(m_k,a_k)}_{\ty},\q\mbox{for all $k\ge1$.}
\end{equation}
Here, 
${\mathcal L}^{(m_k,a_k)}_{\ty}$ is the generalized Cantor fractal string of infinite order defined by Eq.~\eqref{Lty}, while $L_k$ is its total length, given by~\eqref{lengthLty}. We have that
\begin{equation}\label{Lk}
L_k:=|\mathcal L^{(m_k,a_k)}_{\ty}|_1=\exp\Big(\frac1{1-m_ka_k}\Big)-1.
\end{equation}
Since 
$$
\lim_{k\to\ty}m_ka_k=\lim_{k\to\ty} m_k^{1-1/D_k}=(+\ty)^{1-1/D}=0,
$$
we conclude from Eq.~\eqref{Lk} that 
\begin{equation}\label{Le}
\lim_{k\to\ty} L_k=\E-1.
\end{equation}
Let us verify that the fractal string $\mathcal L$, given as the disjoint union of the sequence of bounded fractal strings $({\mathcal L}_k)_{k\ge1}$,
\begin{equation}\label{L}
{\mathcal L}:=\bigsqcup_{k=1}^{\ty}{\mathcal L}_k,
\end{equation}
is well defined and bounded. 
Indeed, we have that
\begin{equation}
|{\mathcal L}|_1=\sum_{k=1}^{\ty}|{\mathcal L}_k|_1=\sum_{k=1}^{\ty}\frac{2^{-k}}{L_k}|{\mathcal L}^{(m_k,a_k)}_{\ty}|_1=
\sum_{k=1}^{\ty}2^{-k}=1,
\end{equation}
where in the next to last equality, we have made use of Eq.~\eqref{Lk}.

\setlength{\unitlength}{1.2cm}
 \linethickness{1pt}
\thinlines

\begin{figure}
\begin{center}
\begin{picture}(6.5,3.5)
\thinlines
\put(0,0){\vector(1,0){6.5}}

\put(.5,-.5){\vector(0,1){3.5}}

\put(.5,0){\circle*{.1}}
\put(6,0){\circle*{.1}}
\put(.27,-.4){\small $0$}
\put(5.95,-.4){\small $1$}

\multiput(4,0)(0,0.4){8}{\circle{.05}}
\put(3.8,-.4){\small $D_1$}

\put(5,0){\circle*{.1}}
\put(4.8,-.4){\small $D$}

\put(5.4,2.7){\small $\Ce$}

\multiput(3.2,0)(0,0.2){15}{\circle{.05}}
\put(3,-.4){\small $D_2$}

\put(2.7,1.4){\small$S_\ty$}

\multiput(2.6,0)(0,0.127){23}{\circle{.05}}

\put(2.4,-.4){\small $D_3$}

\put(2,0){\line(0,1){2.85}}
\put(1.8,-.4){\small $D_\ty$}

\end{picture}
\end{center}
\vskip.5cm
\caption{The set $S_\ty$ of essential singularities (denoted by small circles) of the geometric zeta function $\zeta_{\mathcal L}$, corresponding to the fractal string $\mathcal L$ constructed in the proof of Theorem \ref{main} (see Eqs.~\eqref{L} and \eqref{Lkstring}), accumulates near the vertical line $\{\re s=D_\ty\}$. Here, $D_{\rm par}(\zeta_{\mathcal L})=D_\ty$, $D_{\rm mer}(\zeta_{\mathcal L})=D_1$ and $D(\zeta_{\mathcal L})=\dim{\mathcal L}=D$.}\label{figure}
\end{figure}
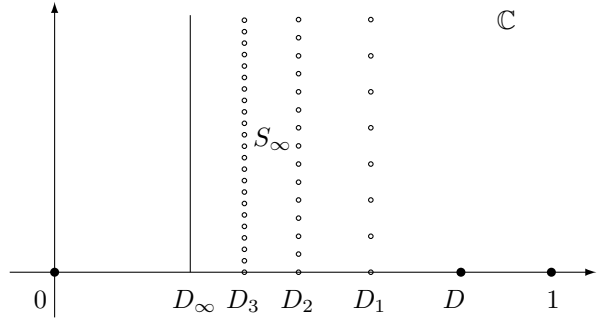

From the definition of the fractal string $\mathcal L$ in \eqref{L} (see also \eqref{Lkstring}), it follows that
\begin{equation}\label{zetaL}
\zeta_{\mathcal L}(s)=\sum_{k=1}^{\ty}\zeta_{{\mathcal L}_k}(s)=\sum_{k=1}^\ty\frac{2^{-ks}}{L_k^s}\zeta_{{\mathcal L}^{(m_k,a_k)}_\ty}(s),
\end{equation}
for all $s\in\Ce$ with $\re s>D_\ty$, except for the set of singularities. All the singularities of $\zeta_{\mathcal L}$, contained in the right half-plane $\{\re s>D_\ty\}$, are essential, and the corresponding set $S_\ty$ of its essential singularities, contained in this same half-plane, coincides with the union over all $k\in\eN$ of the sets of essential singularities of $\zeta_{{\mathcal L}^{(m_k,a_k)}_\ty}$:
\begin{equation}\label{Sty}
S_\ty=\bigcup_{k=1}^\ty\Big(D_k+\frac{2\pi}{\log (1/a_k)}\I\Ze\Big).
\end{equation}
This set $S_{\ty}$ consists of isolated singularities, which means that the geometric zeta function $\zeta_{\mathcal L}$ is paramorphic in the open right half-plane $\{\re s>D_\ty\}$; see Definition~\ref{para}. [That the zeta function $\zeta_{\mathcal L}(s)$ given by \eqref{zetaL} is indeed well defined in the connected open set $\{\re s>D_\ty\}\setminus S_\ty$, is shown in the Appendix (Section~\ref{app}); see Theorem \ref{paraL}.] 

On the other hand, for each arithmetic set $D_k+\frac{2\pi}{\log (1/a_k)}\I\Ze$, the value of 
$\mathbf p_k:=\frac{2\pi}{\log (1/a_k)}$ (which is called the {\em oscillatory period} of the fractal string $\mathcal L_k$; see, e.g., \cite[p.\ 188]{fzf} and \cite{lapidusfrank12}) tends to $0$ as $k\to\ty$, because 
$$
\lim_{k\to\ty}a_k=\lim_{k\to\ty} m_k^{-1/D_k}=(+\ty)^{-1/D}=0.
$$ 
In particular, since $D_k\to D_{\ty}$ as $k\to\ty$, it follows that the set of accumulation points of the set $S_\ty$ coincides with the vertical line $\{\re s=D_{\ty}\}$; see Figure~\ref{figure}. Indeed, assume that $z$ is an arbitrary complex number contained in the vertical line $\{\re s=D_{\ty}\}$. For any connected open neighborhood $N(z)$ of $z$, there are infinitely many essential singularities of $\zeta_{\mathcal L}$ contained in $N(z)\cap\{\re s>D_\ty\}$. This shows that $D_{\ty}$ is equal to the abscissa of paramorphic continuation $D_{\rm par}(\zeta_{\mathcal L})$ of the geometric zeta function $\zeta_{\mathcal L}$.

Since the function $\zeta_{\mathcal L}$ is holomorphic in the open right half-plane $\{\re s>D_1\}$, while $D_1$ is an essential singularity, it follows that the abscissa of meromorphic continuation $D_{\rm mer}(\zeta_{\mathcal L})$ of $\zeta_{\mathcal L}$ is equal to $D_1$. This concludes the proof of the theorem in case $(i)$.
\medskip

{\em Case $(ii)$}: Let $D_\ty<D_1<D$. Let $\mathcal L^{(m',a')}$, where $m'$ is an integer $\ge2$ and $a'\in(0,1/m')$, be a generalized Cantor string such that the abscissa $D(\zeta_{\mathcal L^{(m',a')}})$ of (absolute) convergence of its geometric zeta function $\zeta_{\mathcal L^{(m',a')}}$ is equal to $D$. 
Then, the bounded fractal string $\mathcal L\sqcup \mathcal L^{(m',a')}$, where $\mathcal L$ is the fractal string from step $(i)$, satisfies the desired properties. 
Indeed, we have that $\zeta_{\mathcal L\sqcup \mathcal L^{(m',a')}}(s)=\zeta_{\mathcal L}(s)+\zeta_{\mathcal L^{(m',a')}}(s)$, for all $s\in\Ce$ with $\re s$ sufficiently large. Therefore, $\zeta_{\mathcal L\sqcup \mathcal L^{(m',a')}}$ can be paramorphically continued to the open right half-plane $\{\re s>D_\ty\}$.
\medskip


This completes the proof of the theorem.
\end{proof}

\medskip
The following questions arise naturally in this context:

\medskip


Q1: What does the asymptotics of the tube function of a fractal string look like when $t\to0^+$, in the case when the associated geometric zeta function is paramorphic? For example, in the case of the fractal strings ${\mathcal L}^{(m,a)}_n$ and ${\mathcal L}^{(m,a)}_\ty$ constructed above, as well as for $\mathcal L_\ty$ appearing in Theorem \ref{paraL} of the appendix below.
\medskip

Q2: In the paramorphic case and under suitable polynomial-type growth hypotheses on $\zeta_{\mathcal L}$, is it possible to establish some kind of a tube formula for a fractal string $\mathcal L$ if we know the complex dimensions of $\mathcal L$? 
\medskip

In light of the results of Section~\ref{harm} below, we could ask analogous questions about fractal tube formulas for bounded subsets of $\eR^N$ (for $N\ge2$) and their distance zeta functions instead of for fractal strings and their geometric zeta functions (corresponding to the case when $N=1$, as in \cite[Ch.~8]{lapidusfrank12}). For fractal tube formulas for bounded sets (and, more generally, for relative fractal drums) in $\eR^N$, see \cite[Ch.~5]{fzf} and~\cite{ftf}.
We note that the results about the general fractal tube formulas obtained in \cite{lapidusfrank12} and \cite{fzf,ftf} assume the meromorphicity of a suitable fractal zeta function in a suitable domain of $\Ce$, along with appropriate growth conditions satisfied by this zeta function. Finally, we mention that several results along the lines suggested in question Q2 are provided in~\cite{log}.
\bigskip\hbox{}

\subsection{Harmonic functions and their essential singularities}\label{harm}

We first introduce the notion of an {\em isolated singularity} of a given harmonic function defined on a connected open subset of the Euclidean two-dimensional plane.

\begin{defn}\label{harmonics} Let $U$ be a nonempty connected open subset of the $2$-dimen\-sional plane $\eR^2$. Let $S$ be a set of isolated points of $U$ such that a function $u:U\setminus S\to\eR$ is harmonic in $U\setminus S$. (Observe that the set $U\setminus S$ is necessarily connected as well.) Let $v:U\setminus S\to\eR$ be {\em a} conjugate harmonic function of the given real-valued function $u$ on the connected set $U\setminus S$, meaning that the function $f:U\setminus S\to\Ce$ (here, we identify $U\setminus S$ with the corresponding subset of $\Ce$) defined by
$f(s):=u(x,y)+\I v(x,y)$, where $s:=x+\I y$, is holomorphic in $U\setminus S$.  We then say that a point $(x_0,y_0)\in S$ is an {\em isolated singularity} of $u$ if the corresponding complex number $s_0:=x_0+\I y_0$ is an {\em isolated singularity} of $f$. In particular, we say that a point $(x_0,y_0)\in S$ is an {\em essential singularity} (respectively, {\em pole}) of $u$ if the corresponding complex number $s_0:=x_0+\I y_0$ is an essential singularity (respectively, pole) of $f$.

For a harmonic function $u$ appearing in this definition, we say (in short) that $u$ is {\em paraharmonic} in $U$ if each of its points in $S$ is an isolated singularity of the corresponding holomorphic function $f:U\setminus S\to\Ce$. Or even more succinctly, a harmonic function $u$ said to be {\em paraharmonic} in $U$ if the corresponding complex-valued function $f$ is paramorphic in the set $U$, viewed as a connected open subset of the complex plane.
\end{defn}

It is easy to generate paraharmonic functions from paramorphic functions as shown in the following example.

\begin{example}
	The function $u(x,y)=\re(\E^{1/(x+\I y-x_0-\I y_0)})$ is paramorphic in $\mathbb{R}^2$. Here, $(x_0,y_0)\in\eR^2$ is the only singularity of $u$, and it is essential.
	Of course, the above function is just the real part of the corresponding paramorphic function $f(z)=\E^{1/(z-z_0)}$ discussed in Example \ref{paramex}.
\end{example}

Note that the notion of an isolated singularity (and in particular, of a pole, as well as of an essential singularity) of a harmonic function $u$, introduced in Definition~\ref{harmonics} above, is meaningful since the conjugate harmonic function $v$, defined on a connected open set, is uniquely determined by $u$, up to an additive constant. In light of this observation, adding a constant to the function $f$ does not change the type of any of its isolated singularities.

The following corollary of Theorem \ref{main} shows that there exist explicit real-valued functions $u$ that are paraharmonic in a prescribed open right half-plane $U$ in $\eR^2$, and possessing infinitely many essential singularities, accumulating densely along the boundary~$\pa U$ (which, in this case, is a vertical line).

\begin{cor}\label{harmonic}
Let $U$ be an open right half-plane in $\eR^2$ defined by $U:=\{(x,y)\in\eR^2:x>D_\ty\}$, where $D_\ty\in[0,1)$ is given. Let $D_1$ and $D$ be such that $D_\ty<D_1\le D<1$. Then, there exists an explicitly constructible real-valued function $u$ which is paraharmonic in $U$ $($in the sense of Definition \ref{harmonics}$)$, and is generated by the geometric zeta function $\zeta_{\mathcal L}$ of a bounded fractal string $\mathcal L$ $($i.e., $f=u+\I v=\zeta_{\mathcal{L}}$ in Definition \ref{harmonics}$)$.
Furthermore, the set $S$ of its essential singularities is infinite, contained in the semi-open vertical strip $\{(x,y)\in\eR^2:D_\ty<x\le D_1\}$, and such that the set of accumulation points of $S$ coincides with the vertical line $\{(x,y)\in\eR^2:x=D_\ty\}$, while $u$ is harmonic in the open right half-plane $\{(x,y)\in\eR^2:x>D\}$.

Moreover, the open right half-plane $\{(x,y)\in\eR^2:x>D_\ty\}$ is the maximal right half-plane to which the function $u$ can be paraharmonically extended.
\end{cor}

\begin{proof}
The claim follows immediately from Theorem \ref{main}, by letting $u:=\re \zeta_{\mathcal L}$. The corresponding conjugate harmonic function is $v:=\im \zeta_{\mathcal L}$.
\end{proof}

\begin{remark}
It is possible to construct a class of paraharmonic functions by using paramorphic functions of a simpler type, for example $g(s):=\exp(1/s)$. Here, $g$ is paramorphic on $\Ce$ and $s=0$ is the only isolated singularity of $f$. Furthermore, it is an essential singularity of $f$. If $S=\{a_n:n\in\eN\}$ is any set of isolated points contained in a given connected open set $U\stq\Ce$, then, by using the Weierstrass $M$-test, it is easy to verify that the function
$$
f(s):=\sum_{j=1}^\ty\frac1{n!}g(s-a_n)=\sum_{j=1}^\ty\frac1{n!}\exp(1/(s-a_n))
$$
is paramorphic on $U$. More precisely, $f$ is holomorphic in $U\setminus S$; see Definition~\ref{para}. We can ensure that the set of accumulation points of the set $S$ of isolated singularities of $f$ coincide with the boundary of $U$. However, we do not know if there is a (bounded) fractal string $\mathcal L$ such that $\zeta_{\mathcal L}(s)=g(s-a)$, for all $s\in\Ce$ with $\re s$ sufficiently large, where $a\in(0,1)$ is fixed.
\end{remark}

\section{Essential singularities of distance zeta functions}
\label{sec:5}

Let $A$ be a nonempty bounded set in $\eR^N$, where $N$ is a positive integer, and let $d(x,A):=\inf\{|x-a|:a\in A\}$ denote the Euclidean distance from $x\in\eR^N$ to $A$. Assume that $\d$ is an arbitrary positive real number, and let $A_\d:=\{x\in\eR^N:d(x,A)<\d\}$ be the open $\d$-neighborhood of $A$ in $\eR^N$. The {\em distance zeta function} $\zeta_A$ of the set $A$ is defined by
\begin{equation}\label{zeta_A}
\zeta_A(s):=\int_{A_\d}d(x,A)^{s-N}\dd x,
\end{equation} 
for all $s\in\Ce$ such that $\re s$ is sufficiently large; see \cite{larazu1} or \cite{fzf}.
It is easy to verify that the difference of distance zeta functions corresponding to different values of the parameter $\d>0$ is always an entire function. Hence, the value of the parameter $\d$ is unimportant, since it does not have any influence on the type of any of the isolated singularities of the distance zeta function, considered on any given connected and open subset $U$ of the complex plane. 

We denote by $D(\zeta_A)$ the abscissa of convergence of the Dirichlet-type integral defining $\zeta_A$ on the right-hand side of \eqref{zeta_A}; by definition, this means that $\{\re s>D(\zeta_A)\}$ is the largest right-half plane for which the Lebesgue integral defining $\zeta_A$ in \eqref{zeta_A} is convergent.
Then, according to \cite[Theorems 2.1.11 and 2.1.20]{fzf} we have that
\begin{equation}
	\label{eq:26.1/2}
	D(\zeta_A)=D_A,
\end{equation}
the (upper) Minkowski (or box) dimension of $A$.\footnote{It follows from our hypotheses and the definition of $D_A$ that $0\leq D_A\leq N$.}
Moreover, by analytic continuation, Eq.\ \eqref{zeta_A} continues to hold for all $s\in\Ce$ with $\re s > D(\zeta_A)$; see {\em loc.\ cit}.

We refer the reader to interesting examples of obtained distance zeta functions of various well-known fractal sets, such as the Sierpi\'nski gasket and carpet, which can be found in \cite[Section 3.2]{fzf} as well as in the paper \cite{larazu1}.

\begin{theorem}\label{main2}
Let $N\ge1$ be a fixed but arbitrary integer.
Let $D_{\ty}$, $D_1$ and $D$ be real numbers belonging to the interval $[0,N)$ and such that $D_\ty< D_1\le D$. Then, there exists an explicitly constructible nonempty bounded set $A$ in $\eR^N$ such that the corresponding distance zeta function $\zeta_A$ can be paramorphically extended to the open right half-plane $\{\re s>D_{\ty}\}$ and
\begin{equation}
D_{\rm par}(\zeta_A)=D_\ty,\q
D_{\rm mer}(\zeta_A)=D_1,\q
D(\zeta_A)=D.
\end{equation}
\end{theorem}

The proof of the theorem rests on the following `shift property'. For another related shift property result, see \cite{larazu1} or \cite[Theorem~2.2.32 and Remark~2.2.33]{fzf}.

\begin{lemma}[Shift property of distance zeta functions]\label{shiftl}
Let $\mathcal L=(\ell_j)_{j\in\eN}$ be a bounded fractal string, and let 
\begin{equation}\label{AL}
A_{\mathcal L}:=\{a_k:=\sum_{j=k}^\ty\ell_j,\,\, k\in\eN\}
\end{equation} 
be its canonical {\em geometric realization} contained in $[0,a_1]$. Assume that $\d>\ell_1/2$. Then, for any $N\ge2$,
\begin{equation}\label{shift}
\zeta_{A_{\mathcal L}\times[0,1]^{N-1}}(s)=\frac{2^{N-s}}{s-N+1}\zeta_{\mathcal L}(s-N+1)+\zeta_{A_{\mathcal L}\times\{0\}^{N-1}}(s)
+g(s),
\end{equation}
with $g(s):=\frac{2^{N-s}}{s-N+1}\d^{s-N+1}$, for all $s\in\Ce$ with $\re s$ sufficiently large.\newline
\hbox{}\q In particular, if $S$ is the set of isolated singularities of a paramorphic extension of $\zeta_{\mathcal L}$ to the connected open subset $U\stq \Ce$ and is such that $0\notin S$, then the shifted set $S+(N-1):=\{s+N-1:s\in S\}$ is the set of isolated singularities of the corresponding paramorphic extension of $\zeta_{A_{\mathcal L}\times[0,1]^{N-1}}$. If\, $0\in S$, then
the set of isolated singularities of $\zeta_{A_{\mathcal L}\times[0,1]^{N-1}}$ is $(S+(N-1))\cup\{N-1\}$. 
\end{lemma}

\begin{proof}
Since $\d>\ell_1/2$, we have that $(A_\a)_\d=(-\d,a_1+\d)$. The set $(A_{\mathcal L}\times[0,1]^{N-1})_\d$ contained in $\eR^N$ is connected, and it can be obtained as the union $V_1\cup V_2\cup V_3$ of the following three disjoint subsets of $\eR^N$: 
$$
\begin{gathered}
V_1:=[0,a_1]\times[0,1]^{N-1},\q 
V_3:=(A_{\mathcal L})_\d\cap(\{x_1<0\}\cup\{x_1>a_1\}),\\  
\vbox to4mm{} V_2:=(A_{\mathcal L}\times[0,1]^{N-1})_\d\setminus (V_1\cup V_3).
\vbox to1mm{}
\end{gathered}
$$ 
If we let $A':=A_{\mathcal L}\times[0,1]^{N-1}\st\eR^N$, we have that
\begin{equation}
\zeta_{A'}(s)=\int_{(A')_\d}d(x,A')^{s-N}\dd x\\
=\int_{V_1}+\int_{V_2}+\int_{V_3}.
\end{equation}
The last three integrals are equal to the corresponding three terms on the right hand-side of Eq.~\eqref{shift}. For example, since $$V_1:=\cup_{j=1}^\ty[a_{j+1},a_j]\times[0,1]^{N-1},$$ then
$$
\begin{aligned}
\int_{V_1}d(x,A')^{s-N}\dd x&=\sum_{j=1}^\ty\int_{a_{j+1}}^{a_j}d(x_1,\{a_{j+1},a_j\})^{s-N}\dd x_1\int_0^1\dd x_2\dots\int_0^1\dd x_n\\
&=\sum_{j=1}^\ty2\int_0^{\ell_j/2}\tau^{s-N}\dd\tau=\frac{2^{N-s}}{s-N+1}\sum_{j=1}^\ty\ell_j^{s-N+1}\\
&=\frac{2^{N-s}}{s-N+1}\zeta_{\mathcal L}(s-N+1),
\end{aligned}
$$
for all $s\in\Ce$ with $\re s>N-1$.
We leave to the interested reader the easy (and analogous) verification of the other two equalities.
\end{proof}

\begin{proof}[Proof of Theorem \ref{main2}] We consider the following cases:
\medskip

Case $(i)$: For $N=1$, it suffices to use Theorem \ref{main} with $A:=A_{\mathcal L}$, where $A_{\mathcal L}\st\eR$, defined by Eq.~\eqref{AL}, is the canonical geometric realization of the bounded fractal string $\mathcal L$ appearing in Theorem~\ref{main}.
\medskip

Case $(ii)$: If $N\ge2$, let $A:=A_{\mathcal L}\times[0,1]^{N-1}$; that is, $A$ is the `fractal grill' generated by the set $A_{\mathcal L}$.
Then, the claim follows from Lemma \ref{shiftl}, provided the numbers $D_{\ty}$, $D_1$ and $D$ belong to the interval $[N-1,N)$.
\medskip

Case $(iii)$: In the general case, when $N\ge 2$ and $D_{\ty}$, $D_1$ and $D$ are in $[0,N)$, we first take $N_1$ to be the smallest integer strictly larger than $D_\ty$, and let $D_1'$ be a real number belonging to $(D_\ty,N_1)$, and such that $D_1'\le D_1$. As in case $(ii)$, we first define the set $A_1:=A_{\mathcal L}\times[0,1]^{N_1}$, with the bounded fractal string $\mathcal L$ chosen as in Theorem \ref{main}. Then $D_\ty(\zeta_{A_1})=D_\ty$.
We then let $A'':=A_1\times\{0\}^{N-1-N_1}\st\eR^N$,\footnote{Note that the singularities of the distance and tube fractal zeta functions do not depend on the dimension of the ambient space (see \cite{Res}, along with \cite[Section 4.7]{fzf}); thus, $D_{\rm par}(\zeta_{A_1})=D_{\rm par}(\zeta_{A''})=D_{\infty}$.} and finally, $A:=A''\cup B\cup C$, where the sets $B$ and $C$ are defined so that the corresponding distance zeta functions $\zeta_B$ and $\zeta_C$ can be paramorphically continued to $\{\re s>0\}$, such that $D_1$ and $D$ are then the respective isolated singularities, $D_1$ being the essential singularity of $\zeta_B$ and $D$ the pole of $\zeta_C$.

The set $B$ can be constructed as the fractal grill $B:=C_{\infty}^{(m_1,a_1)}\times[0,1]^d_1$, where $C_{\infty}^{(m_1,a_1)}$ is the canonical geometric realization of the fractal string ${\mathcal L}_{\infty}^{(m_1,a_1)}$, with $d_1:=\lfloor D_1\rfloor$, and the parameters $m_1$ and $a_1$ chosen so that $\log_{1/a_1} m_1=D_1-d_1$. Similarly, the set $C$ can be constructed as the fractal grill $C:=C^{(m,a)}\times[0,1]^d$, by letting $d:=\lfloor D\rfloor$, and the parameters $m$ and $a$ chosen so that $\log_{1/a} m=D-d$. (Here, $C^{(m,a)}$ is the generalized Cantor set introduced in \cite[Definition~3.1.1, p.\ 187]{fzf}, determined by an integer $m\ge2$ and a positive real number $a$ such that~$ma<1$.)

The claim now follows, since $\zeta_A(s)=\zeta_{A''}(s)+\zeta_B(s)+\zeta_C(s)$ for all $s\in\Ce$ with $\re s$ sufficiently large, and hence, $\zeta_A$ can be paramorphically continued to the open right half-plane $\{\re s>D_\ty\}$.
\medskip

This completes the proof of the theorem.
\end{proof}

\section{Appendix}\label{app}

Here, we show that the geometric zeta function defined by Eq.~\eqref{zetaL} is paramorphic in the open right half-plane $\{\re s>D_{\ty}\}$; see Definition \ref{para}. This result is needed in the proof of Theorem~\ref{main}.

\begin{theorem}\label{paraL}
The geometric zeta function $\zeta_{\mathcal L}$ defined by Eq.~\eqref{zetaL} is paramorphic in the open right half-plane $\{\re s>D_{\ty}\}$; that is, $\zeta_{\mathcal L}\in{\rm Par}(\{\re s>D_\ty\})$.
\end{theorem}

The proof of Theorem~\ref{paraL} follows from the following lemma, as will be explained at the end of this appendix.

\begin{lemma}\label{lema}
Let $U$ be any open disk contained in the set $\{\re s>D_\ty\}\setminus S_\ty$ with sufficiently small radius, where the set $S_\ty$ of essential singularities of $\zeta_{\mathcal L}$ is defined by Eq.~\eqref{Sty}, and is such that the Euclidean distance from $U$ to $S_\ty$ is positive.
Then, there exists $\e>0$ such that for all $s\in U$ and for each $k\in\eN$, we have that $|1-m_k\cdot a_k^s|\ge\e$. 
\end{lemma}

\begin{proof}
We consider the following three cases:
\medskip

{\em Case $(a)$}: Let $U$ be as in the statement of the lemma and such that its closure does not intersect any of the vertical lines $\{\re s=D_k\}$, where $k\in\eN$.
We assume that for some $k_0\in\eN$, the set $U$ is placed between the two consecutive lines $\{\re s=D_{k_0}\}$ and $\{\re s=D_{k_0+1}\}$, i.e., in the open vertical strip $\{D_{k_0+1}<\re s< D_{k_0}\}$. In other words,
\begin{equation}
D_{k_0+1}<\inf_{s\in U} \re s\le \sup_{s\in U}\, \re s<D_{k_0}.
\end{equation}
We have that 
$$
\begin{aligned}
|1-m_k\cdot a_k^s|^2&=m_k^2\cdot a_k^{2\re s}-2m_k\cdot a_k^{\re s}\cos \big((\log a_k)(\im s)\big)+1\\
&\ge(1-m_k\cdot a_k^{\re s})^2;
\end{aligned}
$$
so that $|1-m_k\cdot a_k^s|\ge|1-m_k\cdot a_k^{\re s}|$, for all $s\in U$ and $k\in\eN$. Hence, since $a_k=m_k^{-1/D_k}$, we obtain the following inequality:
\begin{equation}\label{ineq}
|1-m_k\cdot a_k^s|\ge|1-m_k^{1-({\re s})/{D_k}}|.
\end{equation}
Now, let us consider the following two subcases:
\medskip

{\em Case $(a1)$}: If $k=1,\dots, k_0$, then, since  $\sup_{s\in U}\re s<D_k$, for any $s\in U$ we have that 
$$
|1-m_k^{1-({\re s})/{D_k}}|=m_k^{1-({\re s})/{D_k}}-1\ge
m_k^{1-(\sup_{s\in U}\re s)/{D_k}}-1>0.
$$ Let 
$$
\e_1:=\min\left\{m_k^{1-(\sup_{s\in U}\re s)/{D_k}}-1:k=1,\dots,k_0\right\}>0.
$$
Then, in light of Eq.~\eqref{ineq}, we have that $|1-m_k\cdot a_k^s|\ge \e_1$, for all $s\in U$ and $1\le k\le k_0$.
\medskip

{\em Case $(a2)$}: If $k\ge k_0+1$, then, since $\inf_{s\in U}\re s>D_k$ for all $\ge k_0+1$, it follows that for any $s\in U$,
$$
\begin{aligned}
|1-m_k^{1-({\re s})/{D_k}}|&=1-m_k^{1-({\re s})/{D_k}}\ge 1-m_{k_0+1}^{1-({\re s})/{D_k}}\\
&\ge  1-m_{k_0+1}^{1-({\re s})/{D_{k_0+1}}}\ge1-m_{k_0+1}^{1-({\inf_{s\in U}\re s})/{D_{k_0+1}}}=:\e_2.
\end{aligned}
$$

By letting $\e:=\min\{\e_1,\e_2\}>0$, we deduce from Eq.~\eqref{ineq} that $|1-m_k\cdot a_k^s|\ge\e$, for all $s\in U$ and $k\in\eN$. This completes the proof of the lemma in case $(a)$.
\bigskip

{\em Case $(b)$}: Assume that the disk $U$ is such that it intersects the vertical line $\{\re s=D_{k_0}\}$, for some $k_0\ge2$, and let $U$ be a disk of sufficiently small radius, so that
$D_{k_0+1}<\inf_{s\in U}\re s\le\sup_{s\in U}\re s<D_{k_0-1}$. Analogously as in case $(a)$, we have that there exists a positive real number $\e_1$ such that $|1-m_k\cdot a_k^s|\ge\e_1$, for all $k\ne k_0$.

When $k=k_0$, there exists a positive constant $\e_2$ such that $h(s):=|1-m_{k_0}\cdot a_{k_0}^s|\ge\e_2$, for all $s\in U$. Indeed, the only zeros of the function $h:\Ce\to [0,+\ty)$ are elements of the arithmetic sequence $S_{k_0}:=D_{k_0}+\frac{2\pi}{\log(1/a_{k_0})}\I\Ze$. 
Since $\ov U$ and $S_{k_0}$ are disjoint (here, $\ov U$ denotes the closure of $U$ in $\Ce$), then $h(s)>0$ for all $s\in\ov U$; so that the continuous function $h=h(s)$ has a strictly positive minimum on the set $\ov U$; that is, $\e_2:=\min_{s\in \ov U}h(s)>0$. 

The claim of the lemma in case $(b)$ follows immediately by choosing $\e:=\min\{\e_1,\e_2\}$.
\medskip

{\em Case $(c)$}: The remaining case when the open disk $U$ is such that $\inf_{s\in U}\re s>D_1$, is treated analogously as in case $(a2)$. 

\medskip

This completes the proof of the lemma.
\end{proof}

\begin{proof}[Proof of Theorem \ref{paraL}]  
Let $U$ be an arbitrary open disk contained in $\{\re s>D_\ty\}\setminus S_\ty$.
Let us show that the geometric zeta function
\begin{equation}\label{zetaLL}
\zeta_{\mathcal L}(s):=\sum_{k=1}^\ty\frac{2^{-ks}}{L_k^s}\zeta_{{\mathcal L}^{(m_k,a_k)}_\ty}(s)
=\sum_{k=1}^\ty\frac{2^{-ks}}{L_k^s}\sum_{n=1}^{\ty}\frac{(1-m_k\cdot a_k^s)^{-n}}{(n!)^{s}}
\end{equation}
is well defined.  

In light of Lemma \ref{lema} and since the sequence $(L_k)_{k\ge1}$ is bounded from below by a positive constant $L$ (see Eq.~\eqref{Le}), we deduce from \eqref{zetaLL} that for all $s\in U$,
\begin{equation} 
\begin{aligned}
|\zeta_{\mathcal{L}}(s)|&\le\sum_{k=1}^\ty\frac{2^{-k\re s}}{L_k^{\re s}}\sum_{n=1}^{\ty}\frac{\e^n}{(n!)^{\re s}}\\
&\le L^{-\re s}\sum_{k=1}^\ty{2^{-k\re s}}\sum_{n=1}^{\ty}\frac{\e^n}{(n!)^{\re s}}=\frac {(2L)^{-\re s}}{1-2^{-\re s}}\sum_{n=1}^{\ty}\frac{\e^n}{(n!)^{\re s}}\\
&\le\frac {(2L)^{-\inf_{s\in U}\re s}}{1-2^{-\sup_{s\in U}\re s}}\sum_{n=1}^{\ty}\frac{\e^n}{(n!)^{\inf_{s\in U}\re s}}<\ty.
\end{aligned}
\end{equation}
Hence, by using the Weierstrass $M$-test, we conclude that the function $\zeta_{\mathcal L}$ is well defined and holomorphic in $\{\re s>D_\ty\}\setminus S_\ty$. By Definition \ref{para}, this means that $\zeta_{\mathcal L}$ is paramorphic in the open right half-plane $\{\re s>D_\ty\}$; that is, $\zeta_{\mathcal L}\in{\rm Par}(\{\re s>D_\ty\})$.\end{proof}

\section{Acknowledgments}

We thank the four anonymous referees for their very
thorough reviewing of this paper, as well as their helpful
constructive criticisms and suggestions, along with new interesting
references we were unaware of.

\end{document}